 \documentclass[smallabstract,smallcaptions]{dccpaper}

\usepackage{epsfig}
\usepackage{citesort}
\usepackage{amsmath}
\usepackage{amssymb}
\usepackage{color}
\usepackage{url}

\usepackage{amsthm}

\usepackage{soul}
\usepackage[utf8]{inputenc}
\usepackage[table]{xcolor}
 \usepackage{cancel} 
\usepackage{float}
\usepackage[margin=1in]{geometry}
\usepackage{graphicx}
\usepackage{xcolor}
\usepackage{amsmath}
\usepackage{amssymb}
\usepackage[sort,nocompress]{cite}
\usepackage[ruled,vlined,lined,boxed,commentsnumbered]{algorithm2e}
\usepackage{todonotes}
\usepackage{verbatim}
\usepackage{subcaption}
\usepackage[normalem]{ulem}
\usepackage{array}
\usepackage{algorithmic}
\usepackage{mathrsfs}
 \usepackage{multirow}
 \usepackage{blindtext}
\usepackage{hyperref}
\usepackage{graphicx}

\usepackage{pgf}
\usepackage{tikz}
\usepackage{tikz-cd}
\tikzcdset{scale cd/.style={every label/.append style={scale=#1},
    cells={nodes={scale=#1}}}}

\usetikzlibrary{arrows,automata,positioning}
\SetKwComment{Comment}{$\triangleright$\ }{}

\newtheorem{theorem}{Theorem}
\newtheorem{lemma}[theorem]{Lemma}
\newtheorem{definition}[theorem]{Definition}

\newlength{\figurewidth}
\newlength{\smallfigurewidth}

\newcommand{\LCP}{\mathsf{LCP}}
\newcommand{\lcp}{\mathop{\mathsf{lcp}}}
\newcommand{\initstate}{s_0}
\newcommand{\vir}[1]{``#1''}

\begin{document}

\title
{\large
\textbf{Computing matching statistics on Wheeler DFAs}
}

\author{%
Alessio Conte$^{1}$, Nicola Cotumaccio$^{2, 3}$, Travis Gagie$^{3}$, Giovanni Manzini$^{1} $,\\ Nicola Prezza$^{4}$ and Marinella Sciortino$^{5}$\\[0.5em]
{\small\begin{minipage}{\linewidth}\begin{center}
$^{1}$ University of Pisa, Italy, \url{alessio.conte@unipi.it}, \url{giovanni.manzini@unipi.it} \\
$^{2}$ GSSI, L'Aquila, Italy, \url{nicola.cotumaccio@gssi.it} \\
$^{3}$ Dalhousie University, Halifax, Canada, \url{nicola.cotumaccio@dal.ca}, \url{travis.gagie@dal.ca}\\
$^{4}$ Ca' Foscari Unversity, Venice, Italy, \url{nicola.prezza@unive.it}\\
$^{5}$ University of Palermo, Italy, \url{marinella.sciortino@unipa.it}\\
\end{center}\end{minipage}}
}


\maketitle
\thispagestyle{empty}
 
 \vspace{10pt} 

\begin{abstract}
Matching statistics were introduced to solve the approximate string matching problem, which is a recurrent subroutine in bioinformatics applications. In 2010, Ohlebusch et al. [SPIRE 2010] proposed a time and space efficient algorithm for computing matching statistics which relies on some components of a compressed suffix tree - notably, the longest common prefix (LCP) array. 
In this paper, we show how their algorithm can be generalized from strings to Wheeler deterministic finite automata. Most importantly, we introduce a notion of LCP array for Wheeler automata, thus establishing a first clear step towards extending (compressed) suffix tree functionalities to labeled graphs.
\end{abstract}

\Section{Introduction}

Given a string $ T $ and a pattern $ \pi $, the classical formulation of the pattern matching problem requires to decide whether the pattern $ \pi $ occurs in the string $ T $ and, possibly, count the number of such occurrences and report the positions where they occur. The invention of the FM-index \cite{ferragina2000}, which is based on the Burrows-Wheeler transform \cite{burrows1994}, opened a new line of research in the pattern matching field. The indexing and compression techniques behind the FM-index deeply rely on the idea of suffix sorting, and over the years have been generalized from strings to trees \cite{ferragina2005}, De Brujin graphs \cite{BOSS, GCSA}, Wheeler graphs \cite{gagie2017, alanko2020} and arbitrary graphs \cite{cotumaccio2021, cotumaccio2022}. In particular, the class of Wheeler graphs is probably the one that captures the intuition behind the FM-index in the simplest way, and indeed the notion of Wheeler order has relevant consequences in automata theory \cite{alanko2020, alanko2021}.

However, in bioinformatics we are not only interested in exact pattern matching, but also in a myriad of variations of the pattern matching problem \cite{gusfield1997}. In particular, \emph{matching statistics} were introduced to solve the approximate pattern matching problem \cite{Chang2005SublinearAS}. A powerful data structure that is able to address the variations of the pattern matching problem at once is the \emph{suffix tree} \cite{weiner1973}. The main drawback of the suffix tree is its space consumption, which is non-negligible both in theory and in practice. As a consequence, the suffix tree has been replaced by the \emph{suffix array} \cite{manber1993}. While suffix arrays do not have all the functionalities of suffix trees, it has been shown that they can be augmented with some additional data structures --- notably, the longest common prefix (LCP) array --- so that it is possible to retrieve the full functionalities of a suffix trees \cite{abouelhoda2004}. All these components can be successfully compressed, leading to the so-called \emph{compressed suffix trees} \cite{sadakane2007}.

The natural question is whether it is possible to provide suffix tree functionalities not only to strings, but also to graphs, and in particular Wheeler graphs. In this paper, we provide a first partial affirmative answer by considering the problem of computing matching statistics. In 2010, Ohlebusch et al. \cite{ohlebusch2010} proposed a time and space efficient algorithm for computing matching statistics which relies on some components of a compressed suffix tree. In this paper, we show how their algorithm can be generalized from strings to Wheeler deterministic finite automata. Most importantly, we introduce a notion of longest common prefix (LCP) array for Wheeler automata, thus establishing an important step towards extending (compressed) suffix tree functionalities to labeled graphs.

\Section{Notation and first definitions}

Throughout the paper, we consider an alphabet $ \Sigma $ and a fixed total order $ \preceq $ on $ \Sigma $. We denote by $ \Sigma^* $ the set of all finite strings on $ \Sigma $ and by $ \Sigma^\omega $ the set of all (countably) infinite strings on $ \Sigma $. The empty word is $ \epsilon $. If $ \alpha \in \Sigma^* $, then $ \alpha^R $ is the reverse string of $ \alpha $. We extend the total order $ \preceq $ from $ \Sigma $ to $ \Sigma^* \cup \Sigma^\omega $ lexicographically. If $ i $ and $ j $ are integers, with $ i \le j $, define $ [i,j] = \{i, i + 1, \dots, j - 1, j \} $. If $ T $ is a string, the $ i $-th character of $ T $ is $ T[i] $, and $ T[i..j] = T[i]..T[j] $.

We will consider deterministic automata $ \mathcal{A} = (Q, E, \initstate, F) $, where $ Q $ is the set of states, $ E \subseteq Q \times Q \times \Sigma $ is the set of labeled edges, $ \initstate \in Q $ is the initial state and $ F \subseteq Q $ is the set of final states. The definition implies that for every $ u \in Q $ and for every $ a \in \Sigma $ there exists at most one edge labeled $ a $ leaving $ u $. Following \cite{alanko2020, alanko2021}, we assume that $ \initstate $ has no incoming edges, and every state is reachable from the initial state; moreover, all edges entering the same state have the same label (\emph{input-consistency}), so that for every $ u \in Q \setminus \{\initstate\} $ we can let $ \lambda (u) $ be the label of all edges entering $ u $. We define $ \lambda (\initstate) = \# $, where $ \# \not \in \Sigma $ is a special character for which we assume $ \# \prec a $ for every $ a \in \Sigma $ (the character $ \# $ is an analogous of the termination character $ \$ $ used for suffix trees and suffix arrays). As a consequence, an edge $ (u', u, a) $ can be simply written as $ (u', u) $, because it must be $ a = \lambda (u) $.

We assume familiarity with the notions of suffix array (SA), Burrows Wheeler transform (BWT), FM-index and backward search~\cite{ferragina2000}.

The \emph{matching statistics} of a pattern $ \pi = \pi[1..m] $ with respect to a string $ T = T[1..n] $ are defined as follows. Assume that $ T[n] = \$ \not \in \Sigma $, where $ \$ \prec a $ for every $ a \in \Sigma $. Determining the matching statistics of $ \pi $ with respect to $ T $ means determining, for $ 1 \le i \le m $, (i) the longest prefix $ \pi' $ of $ \pi[i..m] $ which occurs in $ T $, and (ii) the interval corresponding to the set of all strings starting with $ \pi' $ in the list of all lexicographically sorted suffixes. We can describe (i) and (ii) by means of three values: the length $ \ell_i $ of $ \pi' $, and the endpoints $ l_i $ and $ r_i $ of the interval considered in (ii). For example, let $ T = mississippi\$ $ (see Figure \ref{fig:mississippi}), and $ \pi = stpissi $. For $ i = 1 $, we have $ \pi' = s $, so $ \ell_1 = 1 $ and $ [l_1, r_1] = [9, 12] $ (suffixes starting with $ s $). For $ i = 2 $, we have $ \pi' = \epsilon $, so $ \ell_2 = 0 $ and $ [l_2, r_2] = [1, n] = [1, 12] $ (all suffixes start with the empty string). For $ i = 3 $, we have $ \pi' = pi $, so $ \ell_3 = 2 $, and $ [l_3, r_3] = [7, 7] $ (suffixes starting with $ pi $). For $ i = 4 $, we have $ \pi' = issi $, so $ \ell_4 = 4 $, and $ [l_4, r_4] = [4, 5] $ (suffixes starting with $ issi $). One can proceed analogously for $ i = 5, 6, 7 $.

\begin{figure}
    \centering
        \begin{tabular}{|c|l|c|c|c|c|}
        \hline
            $ i $ & Sorted suffixes & LCP & SA & BWT \\
        \hline
            1 & \$ & & 12 & i  \\
            2 & i\$ & 0 & 11 & p  \\
            3 & ippi\$ & 1 & 8 & s \\
            4 & issippi\$ & 1 & 5 & s  \\
            5 & ississippi\$ & 4 & 2 & m \\
            6 & mississippi\$ & 0 & 1 & \$ \\
            7 & pi\$ & 0 & 10 & p  \\
            8 & ppi\$ & 1 & 9 & i  \\
            9 & sippi\$ & 0 & 7 & s \\
            10 & sissippi\$ & 2 & 4 & s \\
            11 & ssippi\$ & 1 & 6 & i  \\
            12 & ssissippi\$ & 3 & 3 & i \\
        \hline
        \end{tabular}
    \caption{The sorted suffixes of ``mississippi\$'' and the LCP, SA, and BWT arrays.}
    \label{fig:mississippi}
\end{figure} 

\Section{Computing matching statistics for strings}

We will first describe the algorithm by Ohlebusch et al. \cite{ohlebusch2010}, emphasizing the ideas that we will generalize when switching to Wheeler DFAs. The algorithms computes the matching statistics using a number of iterations linear in $ m $ by exploiting the backward search. We start from the end of $ \pi $, and we use the backward search (starting from the interval $ [1, n] $ which corresponds to the set of suffixes prefixed by the empty string) to find the interval of all occurrences of the last character of $ \pi $ in $ T $ (if any). Then, starting from the new interval, we use the backward search to find all the occurrences of the suffix of length 2 of $ \pi $ in $ T $ (if any), and so on. At some point, it may happen that for some $ i \le m + 1 $ we have that $ \pi[i..m] $ occurs in $ T $, but the next application of the backward search returns the empty interval, so that $ \pi[i - 1..m] $ does not occur in $ T $ (the case $i=m+1$ corresponds to the initial setting when $\pi[i..m]$ is the empty string). We distinguish two cases:
\begin{itemize}
    \item (Case 1) If $ l_i = 1 $ and $ r_i = n $, this means that all suffixes of $ T $ are prefixed by $ \pi[i..m] $. This may happen in particular if $ i = m + 1 $: this means that the first backward search has been unsuccessful. We immediately conclude that character $ \pi[i - 1] $ does not occur in $ T $, so $ \ell_{i - 1} = 0 $ and $ [l_{i - 1}, r_{i - 1}] = [1, n] $ (because all suffixes start with the empty string). In this case, in the following iterations of the algorithm, we can simply discard $ \pi[i - 1, m] $: when for $ i' \le i - 2 $ we will be searching for the longest prefix of $ \pi[i', m] $ occurring in $ T $, it will suffice to search for the longest prefix of $ \pi[i', i - 2] $ occurring in $ T $.
    \item (Case 2) If $ l_i > 1 $ or $ r_i < n $, this means that the number of suffixes of $ T $ starting with $ \pi[i..m] $ is less than $ n $. Now, every suffix starting with $ \pi[i..m] $ also starts with $ \pi[i..m - 1] $. If the number of suffixes starting with $ \pi[i..m - 1] $ is equal to the number of suffixes starting with $ \pi[i..m] $, then also $ \pi[i - 1..m - 1] $ does not occur in $ T $. More in general, for $ j \le m - 1 $ we can have that $ \pi[i - 1..j] $ occurs in $ T $  only if the number of suffixes starting with $ \pi[i..j] $ is larger than the number of suffixes starting with $ \pi[i..m] $. Since we are interested in maximal matches, we want $ j $ to be as large as possible: we will show later how to compute the largest integer $ j $ such that the number of suffixes starting with $ \pi[i..j] $ is larger than the number of suffixes starting with $ \pi[i..m] $.
    Notice that $ j $  always exists, because all $ n $ suffixes start with the empty string, but less than $ n $ suffixes start with $ \pi[i..m] $. After determining $j$ we discard $ \pi[j + 1..m] $ (so in the following iterations of the algorithm we will simply consider $ \pi[1..j] $), and we recursively apply the backward search starting from the interval associated with the occurrences of $ \pi [i..j] $ --- we will also see how to compute this interval.
\end{itemize}
Let us apply the above algorithm to $ T = mississippi\$ $ and $ \pi = stpissi $. We start with the interval $ [1, n] = [1, 12] $, corresponding to the empty pattern, and character $ \pi[7] = i $. A backward step yields the interval $ [l_7, r_7] = [2, 5] $ (suffixes starting with $ i $), so $ \ell_7 = 1 $. Now, we apply a backward step from $ [2, 5] $ and $ \pi[6] = s $, obtaining $ [l_6, r_6] = [9, 10] $ (suffixes starting with $ si $), so $ \ell_6 = 2 $. Again, we apply a backward step from $ [9, 10] $ and $ \pi[5] = s $, obtaining $ [l_5, r_5] = [11, 12] $ (suffixes starting with $ ssi $), so $ \ell_5 = 3 $. Again, we apply a backward step from $ [11, 12] $ and $ \pi[4] = i $, obtaining $ [l_4, r_4] = [4, 5] $ (suffixes starting with $ issi $), so $ \ell_4 = 4 $. We now apply a backward step from $ [4, 5] $ and $ \pi[3] = p $, and we obtain the empty interval. This means that no suffix starts with $ pissi $. Notice in Figure \ref{fig:mississippi} that the number of suffixes starting with $ issi $ is equal to the number of suffixes starting with $ iss $ or $ is $, but the number of suffixes starting with $ i $ is bigger. As a consequence, we consider the interval of all suffixes starting with $ i $ --- which is $ [2, 5] $ --- and we apply a backward step with $ \pi[3] = p $. This time the backward step is successful, and we obtain $ [l_3, r_3] = [7, 7] $ (suffixes starting with $ pi $), and $ \ell_3 = 2 $. We now apply a backward step from $ [7, 7] $ and $ \pi[2] = t $, obtaining the empty interval. This means that no suffix starts with $ tpi $. Notice in Figure \ref{fig:mississippi} that the number of suffixes starting with $ p $ is bigger than the number of suffixes starting with $ pi $. The corresponding interval is $ [7, 8] $, but a backward step with $ \pi[2] = t $ is still unsuccessful (so no suffix starts with $ tp $). The number of suffixes starting with $ p $ is smaller than the number of suffixes starting with the empty string (which is equal to $ n = 12 $), so we apply a backward step with $ [1, 12] $ and $ \pi[2] = t $. Since the backward step is still unsuccessful, we conclude that $ \pi[2] = t $ does not occur in $ S $, so $ [l_2, r_2] = [1, n] = [1, 12] $ and $ \ell_2 = 0 $. Finally, we start again from the whole interval $ [1, 12] $, and a backward step with $ \pi[1] = s $ returns $ [l_1, r_1] = [9, 12] $ (suffixes starting with $ s $), so $ \ell_1 = 1 $.

It is easy to see that the number of iterations is linear in $ m $. Indeed, every time we apply a backward step, either we move to the left across $ \pi $ to compute a new matching statistic, or we increase by at least 1 the length of the suffix of $ \pi $ which is forever discarded. This implies that the number of iterations is bounded by $ 2 |\pi| = 2 m $.

We are only left with showing (i) how to compute $ j $ and (ii) the interval of all suffixes starting with $ \pi [i..j] $ in Case 2 of the algorithm. To this end, we introduce the longest common prefix (LCP) array $ \LCP = \LCP[2, n] $ of $ T $. We define $ \LCP [i] $ to be the length of the longest common prefix of the $ (i - 1) $-st lexicographically smallest suffix of $ T $ and the $ i $-th lexicographically smallest suffix of $ T $. In Figure \ref{fig:mississippi} we have $ \LCP[5] = 4 $ because the fourth lexicographically smallest suffix of $ T $ is $ issippi\$ $, the fifth lexicographically smallest suffix of $ T $ is $ ississippi\$ $, and the longest common prefix of $ issippi\$ $ and $ ississippi\$ $ is $ issi $, which has length $ 4 $. Remember that in the example the backward search starting from $ [4, 5] $ (suffixes starting with $ issi $) and $ p $ was unsuccessful, so computing $ j $ means determining the longest prefix of $ issi $ such that the the number of suffixes starting with such a prefix is bigger than $ 2 $. This is easy to compute by using the LCP array: the longest such prefix is the one of length $ \max \{\LCP[4], \LCP[6] \} = \max \{1, 0 \} = 1 $, so that the desired prefix is $ i $. As a consequence, we are only left with showing how to compute the interval of all suffixes starting with the prefix $ i $ --- which is $ [2, 5] $. Notice that in order to compute this interval, it is enough to expand the interval $ [4, 6] $ in both directions as long as the LCP value does not go below~$1$. Since $ \LCP[4] = 1 $, $ \LCP[3] = 1 $, and $ \LCP[2] = 0 $, and we already know that $ \LCP[6] = 0 $, we conclude that the desired interval is $ [2, 5] $. In other words, given a position $ t $, we must be able to compute the biggest integer $ k $ less than $ t $ such that $ \LCP[k] < \LCP[t] $, and the smallest integer $ k $ bigger than $ t $ such that $ \LCP[k] < \LCP[t] $ (in our case, $ t = 4 $). These queries are called PSV (\vir{previous smaller value}) and NSV (\vir{next smaller value}) queries. The LCP array can be augmented in such a way that PSV and NSV queries can be solved efficiently: different space-time trade-offs are possible, we refer the reader to \cite{ohlebusch2010} for details. 



\Section{Matching statistics for Wheeler DFAs}

Let us define Wheeler DFAs \cite{alanko2020}.

\begin{definition}\label{def:wheeler}
Let $ \mathcal{A} = (Q, E, \initstate, F) $ be a DFA. A \emph{Wheeler order} on $ \mathcal{A} $ is a total order $ \le $ on Q such that $ \initstate \le u $ for every $ u \in Q $ and:
\begin{enumerate}
    \item[] \quad  (Axiom 1) If $ u,v \in Q $ and $ u < v $, then $ \lambda(u) \preceq \lambda (v) $.
    \item[] \quad (Axiom 2) If $ (u', u), (v', v) \in E $, $ \lambda (u) = \lambda (v)  $ and $ u < v $, then $ u' < v' $.
\end{enumerate}
A DFA $\mathcal{A}$ is \emph{Wheeler} if it admits a Wheeler order. 


\end{definition}

\begin{figure}
	\centering
	\begin{tikzpicture}[shorten >=1pt,node distance=1.6cm,on grid,auto]
	\tikzstyle{every state}=[fill={rgb:black,1;white,10}]
	
	\node[state, accepting]   (2)                          {$ 2 $};
	\node[state]           (5)  [below of=2]    {$ 5 $};
	\node[state]           (6)  [right of=5]    {$ 6 $};
	\node[state]           (7)  [right of=6]    {$ 7 $};
	\node[state]           (8)  [right of=7]    {$ 8 $};
	\node[state]           (9)  [right of=8]    {$ 9 $};
	\node[state, accepting]           (3)  [above of=7]    {$ 3 $};
	\node[state, accepting]           (4)  [above of=9]    {$ 4 $};
	\node[state]           (10)  [below of=5]    {$ 10 $};
	\node[state]           (11)  [below of=6]    {$ 11 $};
	\node[state]           (12)  [below of=7]    {$ 12 $};
	\node[state]           (13)  [below of=8]    {$ 13 $};
	\node[state]           (14)  [below of=9]    {$ 14 $};
	\node[state]           (15)  [below of=10]    {$ 15 $};
	\node[state]           (16)  [below of=11]    {$ 16 $};	
	\node[state]           (17)  [below of=12]    {$ 17 $};
	\node[state]           (18)  [below of=13]    {$ 18 $};
	\node[state]           (19)  [below of=14]    {$ 19 $};
	\node[state, initial]   (1)  [below of=17]    {$ 1 $};
	\path[->]
	(5) edge node {a}    (2)
	(6) edge node {a}    (3)
	(7) edge node {a}    (3)
	(8) edge node {a}    (3)
	(9) edge node {a}    (4)
	(10) edge node {b}   (5)
	(11) edge node {b}   (6)
	(12) edge node {b}   (7)
	(13) edge node {c}   (8)
	(14) edge node {c}   (9)
	(15) edge node {d}   (10)
	(16) edge node {d}   (11)
	(17) edge node {e}   (12)
	(18) edge node {e}   (13)
	(19) edge node {e}   (14)
	(1) edge [bend left = 50] node {f}   (15)
	(1) edge node {g}   (16)
	(1) edge node {h}   (17)
	(1) edge node {i}   (18)
	(1) edge [bend right = 50] node {l}   (19)
	(2) edge [loop right] node {a}   (2)	
	;
	\end{tikzpicture}
	\caption{A Wheeler DFA. States are numbered according to their positions in the Wheeler order.}\label{fig:wheelerdfa}
\end{figure}
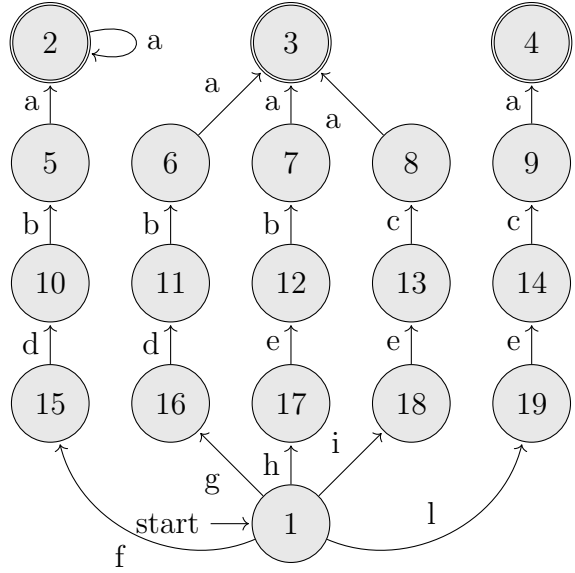

It is immediate to check that this definition is equivalent to the one in \cite{alanko2020}, where it was shown that if a DFA $ \mathcal{A} $ admits a Wheeler order $ \le $, then $ \le $ is uniquely determined (that is, $ \le $ is \emph{the} Wheeler order on $ \mathcal{A} $). In the following, we fix a Wheeler DFA $ \mathcal{A} = (Q, E, \initstate, F) $, where we assume $ Q = \{u_1, \dots, u_n \} $, with $ u_1 < u_2 < \dots < u_n $ in the Wheeler order, and $ u_1 $ coincides with the initial state $\initstate $. See Figure \ref{fig:wheelerdfa} for an example.

We now show that a Wheeler order can be seen of as a permutation of the set of all states playing the same role as the suffix array of a string. In the following, it will be expedient to (conceptually) assume that $ \initstate $ has a self-loop labeled $ \# $ (this is consistent with Axiom 1, because $ \# \prec a $ for every $ a \in \Sigma $). This implies that every state has at least one incoming edge, so for every state $ u_i $ there exists at least one infinite string $ \alpha \in \Sigma^\omega $ that can be read starting from $ u_i $ and following edges in a backward fashion (for example, in Figure \ref{fig:wheelerdfa} for $ u_9 $ such a string is $ cel\#\#\#\dots $). We denote by $ I_{u_i} $ the set of all such strings. Formally:

\begin{definition}
Let $ i \in [1, n] $. For every state $u_i \in Q$ define:
\begin{equation*}
\begin{split}
    I_{u_i} = \{\alpha \in \Sigma^\omega \; | \; \text{there exist integers $ f_1, f_2 , \dots $ in $ [1, n] $ such that (i) $ f_1 = i $,} \\
    \text{(ii) $ (u_{f_{k + 1}}, u_{f_k}) \in E $ for every $ k \ge 1 $ and (iii) $ \alpha = \lambda (u_{f_1}) \lambda (u_{f_2}) \dots $} \}.
\end{split}    
\end{equation*}
\end{definition}
For example, in Figure \ref{fig:wheelerdfa} we have $ I_{u_3} = \{abdg\#\#\#\dots, abeh\#\#\#\dots, acei\#\#\#\dots \} $.

The following lemma shows that the permutation of the states defined by the Wheeler order is the one lexicographically sorting the strings entering each state, just like the permutation defined by the suffix array lexicographically sorts the suffixes of the strings (a suffix is seen as a string \vir{leaving} a text position).

\begin{lemma}\label{lem:suffixinterval}
Let $ i, j \in [1, n] $, with $ i < j $. Let $ \alpha \in I_{u_i} $ and $ \beta \in I_{u_j} $. Then, $ \alpha \preceq \beta $.
\end{lemma}

\begin{proof}
Let $ f_1, f_2 , \dots $ in $ [1, n] $ be such that (i) $ f_1 = i $, (ii) $ (u_{f_{k + 1}}, u_{f_k}) \in E $ for every $ k \ge 1 $ and (iii) $ \alpha = \lambda (u_{f_1}) \lambda (u_{f_2}) \dots $. Analogously, let $ g_1, g_2 , \dots $ in $ [1, n] $ be such that (i) $ g_1 = j $, (ii) $ (u_{g_{k + 1}}, u_{g_k}) \in E $ for every $ k \ge 1 $ and (iii) $ \beta = \lambda (u_{g_1}) \lambda (u_{g_2}) \dots $. Let $ \alpha \not = \beta $. We must prove that $ \alpha \prec \beta $. Let $ p \ge 1 $ be the smallest integer such that the $ p $-th character of $ \alpha $ is different than the $ p $-th character of $ \beta $. In other words, we know that $ \lambda (u_{f_1}) = \lambda (u_{g_1}) $, $ \lambda (u_{f_2}) = \lambda (u_{g_2}) $, $ \dots $, $ \lambda (u_{f_{p - 1}}) = \lambda (u_{g_{p - 1}}) $, but $ \lambda (u_{f_p}) \not = \lambda (u_{g_p}) $. We must prove that $ \lambda (u_{f_p}) \prec \lambda (u_{g_p}) $. Since $ \lambda (u_{f_1}) = \lambda (u_{g_1}) $  $ f_1 = i < j = g_1 $, and $ (u_{f_{2}}, u_{f_1}), (u_{g_{2}}, u_{g_1}) \in E $, from Axiom 2 we obtain $ f_2 < g_2 $. Since $ \lambda (u_{f_2}) = \lambda (u_{g_2}) $, $ f_2 < g_2 $, and $ (u_{f_{3}}, u_{f_2}), (u_{g_{3}}, u_{g_2}) \in E $, from Axiom 2 we obtain $ f_3 < g_3 $. By iterating this argument, we conclude $ f_p < g_p $. By Axiom 1, we obtain $ \lambda (u_{f_p}) \preceq \lambda (u_{g_p}) $. Since $ \lambda (u_{f_p}) \not = \lambda (u_{g_p}) $, we conclude $ \lambda (u_{f_p}) \prec \lambda (u_{g_p}) $.
\end{proof}

If we think of a string as a labeled path, then the suffix array sorts the strings that can be read from each position by moving forward (that is, the suffixes of the string), while the Wheeler order sorts the strings that can be read from each position by moving backward towards the initial state. The underlying idea is the same: the forward vs backward difference is only due to historical reasons~\cite{gagie2017}. To compute the matching statistics on Wheeler DFA we reason as in the previous section replacing backward search with the \emph{forward search}~\cite{gagie2017} defined as follows: given an interval $ [i, j] $ in $ [1, n] $ and $ a \in \Sigma $, find the (possibly empty) interval $ [i', j'] $ in $ [1, n] $ such that a state $ v_{k'} $ is reachable from some state $ v_k $, with $ i \le k \le j $, through an edge labeled $ a $, if and only if $ i' \le k' \le j' $ (this easily follows by using the axioms of Definition \ref{def:wheeler}). 
For a constant size alphabet, given $ [i, j] $ and $ a $ then $ [i', j'] $ can be determined in constant time. Given a string $ \pi \in \Sigma^* $, if we start from the whole set of states and repeatedly apply the forward search we reach the set of all states $ u_i $ for which there exists $ \alpha \in I_{u_i} $ prefixed by $ \pi^R $; this is an interval with respect to the Wheeler order: in the following we call this interval $ T(\pi) $.

Because of the forward vs backward difference the problem of matching statistics will be defined in a symmetrical way on Wheeler DFAs. Given a pattern $ \pi = \pi[1..m] $, for every $ 1 \le i \le m $ we want to determine (i) the longest suffix $ \pi' $ of $ \pi[1..i] $ which occurs in the Wheeler DFA $ \mathcal{A} $ (that is, that can be read somewhere on $ \mathcal{A} $ by concatenating edges), and $ (ii) $ the endpoints of the interval $ T(\pi') $.  


Broadly speaking, we can apply the same idea of the algorithm for strings, but in a symmetrical way. We start from the \emph{beginning}  of $ \pi $ (not from the end of $ \pi) $, and initially we consider the whole set of states. We repeatedly apply the \emph{forward} search (not the backward search), until the forward search returns the empty interval for some $ i \ge 0 $. This means that $ \pi[1..i + 1] $ does not occur in $ \mathcal{A} $. Then, if $T(\pi[1..i])$ is the whole set of states, we conclude that the character $ \pi[i + 1] $ labels no edge in the graph. Otherwise, we must find the smallest $ j $ such that $ T(\pi[1..i]) $ is strictly contained in $ T(\pi[j..i]) $ (that is, we must determine the longest suffix $ \pi[j..i] $ of $ \pi[1..i] $ which reaches more states than $ \pi[1..i] $). Then we must determine the endpoints of the interval $ T(\pi[j..i]) $ so that we can go on with the forward search.

The challenge now is to find a way to solve the same subproblems that we identified in Case 2 of the algorithm for strings. In other words, we must find a way to determine $ j $ and find the endpoints of the interval $ T(\pi[j..i]) $. We will show that the solution is not as simple as the one for the algorithm on strings.

\Section{The LCP array and matching statistics for Wheeler DFAs} 


We start observing that $ I_{u_i} $ may be an infinite set. For example, in Figure \ref{fig:wheelerdfa}, we have
\begin{equation*}
    I_{u_2} = \{aaaaa\dots, abdf\#\#\#\dots, aabdf\#\#\#\dots, aaabdf\#\#\#\dots, \dots \}.
\end{equation*}

In general, an infinite set of (lexicographically sorted) strings in $ \Sigma^\omega $ need not admit a minimum or a maximum. For example, the set $ \{baaaa\dots, abaaa\dots, aabaa\dots, aaaba\dots \} $ does not admit a minimum (but only the \emph{infimum} string $ aaaaa\dots $). Nonetheless, Lemma \ref{lem:suffixinterval} implies that each $ I_{u_i} $ admits both a minimum and a maximum. For example, the minimum is obtained as follows. Let $ f_1 = i $, and for every $ k \ge 1 $, recursively let $ f_{k + 1} $ be the smallest integer in $ [1, n] $ such that $ (u_{f_{k + 1}}, u_{f_k}) \in E $. Then, the minimum of $ I_{u_i} $ is $ \lambda (u_{f_1}) \lambda (u_{f_2}) \dots $, and analogously one can determine the maximum.

In the following, we will denote the minimum and the maximum of $ I_{u_i} $ by $ \min_i $ and $ \max_i $, respectively (for example, in Figure \ref{fig:wheelerdfa} we have $ \min_2 = aaaaa\dots $, and $ \max_2 = abdf\#\#\#\dots $). Lemma \ref{lem:suffixinterval} implies that:
\begin{equation*}
    \mathrm{min}_1 \preceq \mathrm{max}_1 \preceq \mathrm{min}_2 \preceq \mathrm{max}_2 \preceq \dots \preceq \mathrm{max}_{n - 1} \preceq \mathrm{min}_n \preceq \mathrm{max}_n.
\end{equation*}
This suggests to generalize the LCP array as follows. Given $ \alpha, \beta \in \Sigma^* \cup \Sigma^\omega $, let $ \lcp(\alpha, \beta) $ be the length of the longest common prefix of $ \alpha $ and $ \beta $ (if $ \alpha = \beta \in \Sigma^\omega $, define $ \lcp(\alpha, \beta) = \infty $).

\begin{definition}
The \emph{LCP-array} of a Wheeler automaton $ \mathcal{A} $ is the array $ \LCP_\mathcal{A} = \LCP_\mathcal{A}[2, 2n] $ which contains the following $ 2n - 1 $ values in this order: $ \lcp(\min_1, \max_1) $, $ \lcp(\max_1, \min_2) $, $ \lcp(\min_2, \max_2) $, $ \dots $, $ \lcp(\max_{n - 1}, \min_n) $, $ \lcp (\min_n, \max_n) $.
\end{definition}

From the above characterization of $ \min_i $ and $ \max_i $, one can prove that for every entry either $\LCP_\mathcal{A}[i]=\infty$ or $\LCP_\mathcal{A}[i]<3n$ (it follows from Fine and Wilf Theorem \cite{FineWilf65,MantaciRRS07}), and  one can design a polynomial time algorithm to compute $ \LCP_\mathcal{A}$.


Unfortunately, the array $ \LCP_\mathcal{A} $ alone is not sufficient for computing matching statistics. Assume that $ T(\pi) = \{u_r, u_{r + 1}, \dots, u_{s - 1}, u_s \} $, and that when we apply the forward search by adding a character $ c $, we obtain $ T(\pi c) = \emptyset $. We must then determine the largest suffix $ \pi' $ of $ T(\pi) $ such that $ T(\pi) $ is \emph{strictly} contained in $ T(\pi') $. Suppose that \emph{every} string in $ I_{u_r} $ is prefixed by $ \pi^R $, and \emph{every} string in $ I_{u_s} $ is prefixed by $ \pi^R $. In particular, both $ \min_r $ and $ \max_s $ are prefixed by $ \pi^R $. In this case, we can proceed like in the algorithm for strings: the desired suffix $ \pi' $ is the one having length $ \max \{\lcp(\max_{r - 1}, \min_r), \lcp(\max_{s}, \min_{s + 1}) \} $, which can be determined using $ \LCP_\mathcal{A} $.
However, in general, even if \emph{some} string in $ I_{u_r} $ must be prefixed by $ \pi^R $, the string $ \min_r $ need not be prefixed by $ \pi^R $, and similarly $ \max_s $ need not be prefixed by $ \pi^R $. The worst-case scenario occurs when $ r = s $. Consider Figure \ref{fig:wheelerdfa}, and assume that $ \pi = heba $. Then, we have $ r = s = 3 $ (note that $ abeh\#\#\#\dots $ is a string in $I_{u_3} $ prefixed by $ \pi^R $). However, both $ \min_3 = abdg\#\#\#\dots $, and $ \max_3 = acei\#\#\#\dots $, are not prefixed by $ \pi^R $. Notice that $ \lcp(\max_2, \min_3) = 3 $ and $ \lcp(\max_3, \min_4) = 3 $, but $ \pi' $ is not the suffix of length 3 of $ \pi $. Indeed, since $ \min_3 $ is only prefixed by the prefix of $ \pi^R $ of length $ 2 $, and $ \max_3 $ is only prefixed by the prefix of $ \pi^R $ of length $ 1 $, we conclude that it must be $ |\pi'| = 2 $. In general, the desired suffix $ \pi' $ is the one having length
$|\pi'|$ given by:
\begin{equation}\label{eq:pi'}
    \max \left\{\,
    \mathrm{min} \{\lcp(\mathrm{max}_{r - 1}, \mathrm{min}_r),\!\lcp(\mathrm{min}_r, \pi^R) \}, 
    \mathrm{min} \{\lcp(\pi^R, \mathrm{max}_s),\!\lcp(\mathrm{max}_{s}, \mathrm{min}_{s + 1}) \}\, \right\}.
\end{equation}
The above formula shows that, in order to compute $\pi'$, in addition to $ \LCP_\mathcal{A} $ it suffices to know the values $ \lcp(\min_r, \pi^R) $ and $ \lcp(\pi^R, \max_s) $ ($\pi'$ is a suffix of $\pi$, so it is determined by its length).  
We now show how our algorithm can efficiently maintain the current pattern $\pi$, the set $ T(\pi) = \{u_r, u_{r + 1}, \dots, u_{s - 1}, u_s \} $ and the values $ \lcp(\min_r, \pi^R) $ and $ \lcp(\pi^R, \max_s) $ during the computation of the matching statistics. We assume that the input automaton is encoded with the rank/select data structures supporting the execution of a step of forward search in $ O(\log |\Sigma|) $ time, see \cite{gagie2017} for details. In addition, we will use the following result.


\begin{lemma}\label{lem:rmq}
    Let $ A[1, n] $ be a sequence of values over an ordered alphabet $ \Sigma $. Consider the following queries: (i) given $ i, j \in [1..n] $, compute the minimum value in $ S[i..j] $, and (ii) given $ t \in [1..n] $ and $ c \in \Sigma $, determine the biggest $ k < t $ (or the smallest $ k > t $) such that $ A[k] < c $. Then, $ A $ can be augumented with a data structure of $ 2n + o(n) $ bits such that query (i) can be answered in constant time and query (ii) can be answered in $ O(\log n) $ time.
\end{lemma}

\begin{proof}
    There exists a data structure of $ 2n + o(n) $ bits that allows to solve range minimum queries in constant time \cite{fischer2010}, so using $ A $ we can solve queries (i) in constant time. Now, let us show how to solve queries (ii). Let $ f_1 $ be the answer of query (i) on input $ i = \lceil t/2 \rceil $ and $ j = t - 1 $. If $ f_1 < c $, then we must keep searching in the interval $ [\lceil t/2 \rceil, t - 1] $, otherwise, we must keep searching in the interval $ [1, \lceil t/2 \rceil - 1] $. In other words, we can answer a query (ii) by means of a binary search on $ [1, t - 1] $, which takes $ O(\log t) $ (and so $ O(\log n) $) time.
\end{proof}


Notice that query (ii) can be seen as a variant of PSV and NSV queries. In the following, we assume that the array $ \LCP_\mathcal{A} $ has been augmented with the data structure of Lemma~\ref{lem:rmq}.


At the beginning we have $ \pi = \epsilon $, so $ T(\epsilon) = \{1, 2, \ldots, n \} $ and trivially $ \lcp(\min_r, \pi^R) = \lcp(\pi^R, \max_s) = 0 $. At each iteration we perform a step of forward search computing $ T(\pi c) $ given $ T(\pi) $; then we distinguish two cases according to whether  $ T(\pi c) $ is empty or not. 

\smallbreak
\noindent{\bf Case 1.} $ T(\pi c) = \{u_{r'}, u_{r' + 1}, \dots, u_{s' - 1}, u_{s'} \} $ is not empty. In that case $\pi c$ will become the pattern at the next iteration. Since we already have  $ T(\pi c) $ we are left with the task of computing $ \lcp(\min_{r'}, c \pi^R) $ and $ \lcp(c \pi^R, \max_{s'}) $. We only show how to compute $ \lcp(\min_{r'}, c\pi^R) $, the latter computation being analogous. Let $ k $ be the smallest integer in $ [1, n] $ such that $ (u_k, u_{r'}) \in E $. Notice that we can easily compute $ k $ by means of standard \textrm{rank}/\textrm{select} operations on the compact data structure used to encode $ \mathcal{A} $. Since $ u_{r'} \in T(\pi c) $, it must be $ k \le s $. Moreover, the characterization of $ \min_{r'} $ that we described above implies that $ \min_{r'} = c \min_k $, hence $ \lcp(\min_{r'}, c\pi^R) = \lcp(c \min_k, c\pi^R) = 1 + \lcp(\min_k, \pi^R) $. To compute $ \lcp(\min_k, \pi^R) $ we distinguish two subcases:
\begin{enumerate}
    \item[$a)$] $ k > r $, hence $ r < k \le s $. Since $ u_r, u_s \in T(\pi) $, there exist $ \alpha \in I_{u_r} $ and $ \beta \in I_{u_s} $ both prefixed by $ \pi^R $. But $ \alpha \preceq \max_r \preceq \min_k \preceq \min_s \preceq \beta $, so $ \min_k $ is also prefixed by $ \pi^R $, and we conclude $ \lcp(\min_k, \pi^R) = |\pi| $.
    \item[$b)$] $ k \le r $. In this case, we have $ \min_k \preceq \max_k \preceq \min_{k + 1} \prec \max_{k + 1} \preceq \dots \preceq \min_r \prec \pi^R $, and therefore $\lcp(\mathrm{min}_k, \pi^R)$ is equal to
    \begin{equation*}
     \mathrm{min} \{\lcp(\mathrm{min}_k, \mathrm{max}_k), \lcp(\mathrm{max}_k, \mathrm{min}_{k + 1}), \lcp(\mathrm{min}_{k + 1}, \mathrm{max}_{k + 1}), \dots, \lcp(\mathrm{min}_r, \pi^R) \}.
    \end{equation*}
    With the above formula we can compute $ \lcp(\min_k, \pi^R) $ using query (i) of Lemma~\ref{lem:rmq} over the range $ \LCP_\mathcal{A}[2k,2r-1] $ and the value $ \lcp(\min_r, \pi^R) $.
\end{enumerate}


\smallbreak
\noindent{\bf Case 2.} $ T(\pi c)$ is empty. In this case at the next iteration the pattern will be largest suffix $ \pi' $ of $ \pi$ such that $ T(\pi) $ is \emph{strictly} contained in $ T(\pi') = \{u_{r''},  \dots,  u_{s''} \}$. We compute $|\pi'|$ using~\eqref{eq:pi'}; if $ |\pi'| > \lcp(\mathrm{min}_r, \pi^R)$  we set $r'' = r$, otherwise we apply query (ii) of Lemma~\ref{lem:rmq} to find the rightmost entry $r''$ in $ \LCP_\mathcal{A}[2,2r - 1] $ smaller than $|\pi'|$. 
Computing $s''$ is analogous.



Given $ T(\pi') = \{u_{r''}, u_{r'' + 1}, \dots, u_{s'' - 1}, u_{s''} \} $, where $ r'' \le r $, $ s \le s'' $, and at least one inequality is strict, we want to compute $ \lcp(\min_{r''}, (\pi')^R) $ and $ \lcp((\pi')^R, \max_{s''}) $. We only consider $ \lcp(\min_{r''}, (\pi')^R) $, the latter computation being analogous. We distinguish two subcases:
\begin{enumerate}
    \item[$a)$] $ r'' = r $. Then $ \lcp(\min_{r''}, (\pi')^R) = \lcp(\min_{r}, (\pi')^R) = \min \{\lcp(\min_{r}, \pi^R), |\pi'| \} $.
    \item[$b)$] $ r'' < r $. In particular, since $ u_{r''} $ is the left endpoint of $ T(\pi') $ and $ |T(\pi')| \ge 2 $, one can prove like in Case $1a)$ that $ \max_{r''} $ is prefixed by $ (\pi')^R $. We immediately conclude that $ \lcp(\min_{r''}, (\pi')^R) = \min \{ \lcp(\min_{r''}, \max_{r''}), |\pi'| \} $, which can be immediately computed since $ \lcp(\min_{r''}, \max_{r''}) $ is a value stored in $ \LCP_\mathcal{A} $.
\end{enumerate}

\noindent
We can summarize the above discussion as follows.

\begin{theorem}
Given a Wheeler DFA $\mathcal A$, there exists a data structure occupying $O(|\mathcal A|)$ words which can compute the pattern matching statistics of a pattern $P$ in time $O(|P| \log |\mathcal A|)$.
\end{theorem}


\paragraph*{Funding} TG funded by National Institutes of Health (NIH) NIAID (grant no.\ HG011392), the National Science Foundation NSF IIBR (grant no.\ 2029552) and a Natural Science and Engineering Research Council (NSERC) Discovery Grant (grant no.\ RGPIN-07185-2020). GM funded by the Italian Ministry of University and Research (PRIN 2017WR7SHH). MS funded by the INdAM-GNCS Project (CUP\_E55F22000270001).
NP funded by the European Union (ERC, REGINDEX, 101039208).
Views and opinions expressed are however those of the author(s) only and do not necessarily reflect those of the European Union or the European Research Council. Neither the European Union nor the granting authority can be held responsible for them.

\Section{References}
\bibliographystyle{IEEEbib}
\bibliography{refs}

\end{document}